\documentclass[10pt,svgnames,doc]{apa7}
\usepackage{amsthm,amsmath,graphicx,hyperref,lmodern,amsfonts,amssymb,amsfonts,multirow,bbm,breqn,tabularx}
\usepackage{upgreek,booktabs,mathtools,enumerate,bm,multicol,multirow}

\usepackage{xcolor}

\hypersetup{
  colorlinks,
  linkcolor={FireBrick},
  citecolor={MidnightBlue}
}
  \usepackage[backend=biber,style=apa]{biblatex}
\newtheorem{theorem}{Theorem}
\newtheorem{corollary}{Corollary}[theorem]
\newtheorem{lemma}[theorem]{Lemma}
\theoremstyle{definition}

\newtheorem{remark}{Remark}
\newcommand\Tau{\mathrm{T}}

  \theoremstyle{definition}
  \newtheorem{assumption}{Assumption}

\title{An exact unbiased semi-parametric maximum quasi-likelihood framework which is complete in the presence of ties}
\author{\href{mailto:ljrhurley@gmail.com}{Landon Hurley}}

\shorttitle{Hurley}
\authorsaffiliations{Unaffiliated}
\leftheader{Hurley}
\abstract{This paper introduces a novel quasi-likelihood extension of the generalised Kendall \(\tau_{a}\) estimator, together with an extension of the Kemeny metric and its associated covariance and correlation forms. The central contribution is to show that the U-statistic structure of the proposed coefficient \(\tau_{\kappa}\) naturally induces a quasi-maximum likelihood estimation (QMLE) framework, yielding consistent Wald and likelihood ratio test statistics. The development builds on the uncentred correlation inner-product (Hilbert space) formulation of Emond and Mason (2002) and resolves the associated sub-Gaussian likelihood optimisation problem under the \(\ell_{2}\)-norm via an Edgeworth expansion of higher-order moments.  The Kemeny covariance coefficient \(\tau_{\kappa}\) is derived within a novel likelihood framework for pairwise comparison-continuous random variables, enabling direct inference on population-level correlation between ranked or weakly ordered datasets. Unlike existing approaches that focus on marginal or pairwise summaries, the proposed framework supports sample-observed weak orderings and accommodates ties without information loss. Drawing parallels with Thurstone's Case V latent ordering model, we derive a quasi-likelihood-based tie model with analytic standard errors, generalising classical U-statistics. The framework applies to general continuous and discrete random variables and establishes formal equivalence to Bradley-Terry and Thurstone models, yielding a uniquely identified linear representation with both analytic and likelihood-based estimators.}
\begin{document}
\maketitle

A novel quasi-likelihood extension of the generalised \(\tau_{a}\) estimator \parencite{wedderburn1974,heyde1997} extension to the Kemeny metric and its covariance form are introduced in this paper. Primary focus is upon showing that the U-statistic properties of \(\tau_{\kappa}\) equivalently establish quasi-maximum likelihood estimator (QMLE) principles, with consistent Wald and likelihood ratio statistic estimators. This approach builds upon the original uncentred correlation inner-product (Hilbert space) framework of \textcite{emond2002}, and solves the sub-Gaussian likelihood optimisation problem upon the \(\ell_{2}\)-norm by Edgeworth expansion upon the higher-order moments. The properties of the original closed form expressions upon both discrete and continuous i.i.d. data are addressed in Hurley (2025a), which introduces proofs of U-statistic properties and consistent studentisation of the null hypothesis distribution. The Kemeny covariance (and correlation) coefficient \(\tau_{\kappa}\) in this work is developed upon a novel likelihood framework constructed for random variables which are continuous upon their comparisons between each pair of observations. The main contribution of this work is a novel quasi-likelihood based approach upon sample estimated weak-orderings, which allows us to explore the estimator and its properties. While previous methods have often focused on pairwise comparisons or individual marginal distributions, our approach allows for a direct estimation of the population-level correlation between ranked datasets, which is often the primary quantity of interest in ranking-based problems.

This work expands upon a common point in non-parametric, or rank-based, methods that ignoring ties causes information to be lost. Our framework here derives, in a structure similar to Case V of Thurstone's approaches to latent ordering, a likelihood-based tie modelling framework with analytic standard errors which generalises U-statistics. Although rankings serve as a natural example in this model, the random variables \(X\) and \(Y\) are not restricted to rankings alone. Instead, they represent general random variables that can be continuous or discrete. This allows for the application of weak-order scores to a broader range of data types, including continuous measurements, count data, or binary outcomes. The weak-order score function is not specific to rankings but is a general tool for mapping any random variable to a scale that reflects relative comparison among observations.

\subsection{Preliminaries and Notation}

Let \(X = (X_1, \dots, X_N)\) and \(Y = (Y_1, \dots, Y_N)\) be i.i.d. real-valued random variables. For each marginal, define the hollow score-matrix mapping as

\begin{equation}
\label{eq:score_matrix}
C_{kl}(X) = 
\begin{cases} 
+1 & \text{if } X_{k} \ge X_{l}, \\
-1 & \text{if } X_{k} < X_{l},\\
 0 & \text{if } k = l,
\end{cases} \quad
C_{kl}(Y) = 
\begin{cases} 
+1 & \text{if } Y_k \ge Y_l, \\
-1 & \text{if } Y_k < Y_l,\\
0  & \text{if } k = l.
\end{cases}
\end{equation}

Let \(\tilde{\kappa}^X_{kl}\) and \(\tilde{\kappa}^Y_{kl}\) denote the centred score matrices. Note that the off-diagonal elements of the kernel-product matrix \(Z_{kl} = \tilde{\kappa}^X_{kl}\tilde{\kappa}^Y_{kl}\) are \emph{not independent}, due to shared rows and columns. Only the \(N\) marginal observations are independent.

The centred score-matrix upon an \(N \times 1\) i.i.d. column vector is defined as:
\begin{equation}
\label{eq:centred_kappa}
\tilde{\kappa}(X)_{kl} {:=} \kappa(X)_{kl} - \bar{C}_{k\cdot}^{X} - \bar{C}_{\cdot l}^{X} + \bar{C}_{\cdot \cdot}^{X},
\end{equation}
where \( \bar{C}_{k\cdot}^{X} =\)  \((N-1)^{-1} \sum_{l=1}^N \kappa(X)_{kl}\), \( \bar{C}_{\cdot l}^{X} =\) \((N-1)^{-1} \sum_{k=1}^N \kappa(X)_{kl}\), and \( \bar{C}_{\cdot \cdot}^{X} =\) \((N^{2} - N)^{-1} \sum_{k=1}^N \sum_{l=1}^N \kappa(X)_{kl}\) are the row, column, and grand means of the score matrix \( C_{kl}^{X}\), where sums are taken over all indices \( k \) and \( l \), and the diagonal elements are explicitly set to zero, consistent with the definition of the (centred) score-matrix as a hollow matrix.

\section{Maximum Likelihood Estimation via Quasi-Likelihood}

\begin{assumption}[Main Conditions for the \(\kappa\)-Quasi-Likelihood]
\label{cond:main_conditions}

The data \(\{(X_{n},Y_{n})\}_{i=1}^N\) satisfy the following conditions.

\textit{(A)}  
The observations are i.i.d. drawn from a joint distribution on \(\mathbb{R}^{P}\times\mathbb{R}\) with orderable (possibly discrete) margins for \(Y\) and arbitrary (possibly discrete) margins for \(X\). Ties in \(X,Y\) are valid observations permitted and assigned positive probability (in contradiction to \(\tau_{a},\tau_{b}\)).

\textit{(B)}  
The design matrix \(X\in\mathbb{R}^{N\times P}\) has full column rank when expressed in
pairwise-contrast form: the set of differences \(\{X_{n}-X_{n^{\prime}} : n \neq n^{\prime}\}\) spans \(\mathbb{R}^{P}\). In particular, for every non-zero \(v\in\mathbb{R}^P\),
\[
\exists\, (n,n^{\prime}) \text{ such that }
(X_{n}-X_{n^{\prime}})^\top v \neq 0.
\]

\textit{(C)} The linear predictor
\[
\eta_{kl}= (X_{n}-X_{n^{\prime}})^\top \theta
\]
satisfies \(\mid\eta_{kl}\mid<1\) for all \(n \neq n^{\prime}\). Equivalently, the parameter space is the convex set
\[
\Theta = \bigl\{ \theta\in\mathbb{R}^{P} : 
|(X_{n}-X_{n^{\prime}})^\top\theta|<1\ \text{for all } n \neq n^{\prime} \bigr\}.
\]

\textit{(D)}  
The quasi-likelihood score
\[
S(\theta)=
\sum_{n \neq n^{\prime}}
\frac{2\,\eta_{kl}}{1-\eta_{kl}^2}(X_{n}-X_{n^{\prime}})
\]
is continuously differentiable on \(\Theta\); the Hessian
\[
H(\theta)
=
\sum_{n \neq n^{\prime}}
2\,\frac{1+\eta_{kl}^2}{(1-\eta_{kl}^2)^2}
(X_{n}-X_{n^{\prime}})(X_{n}-X_{n^{\prime}})^\top
\]
is positive definite for all \(\theta\in\Theta\).

\textit{(E)}  
The variance of the H\'ajek projection \parencite{vaart1998} of \(Z_{kl}\) is positive:
\[
\mathbb{V}
\Big(
\mathbb{E}[Z_{kl}\mid X_{n},Y_{n}]
\Big)
>0.
\]

\end{assumption}

\begin{lemma}[Quasi-Likelihood MLE]
\label{lemma:quasi_mle}
Let \(X_1, \dots, X_N\) and \(Y_1, \dots, Y_N\) be i.i.d. random variables. Define the normalised correlation coefficient
\[
\hat{\tau}_{\kappa} = 
\frac{\sum_{k\neq l} \tilde{\kappa}^X_{kl}\tilde{\kappa}^Y_{kl}}
{\sqrt{\sum_{k\neq l} (\tilde{\kappa}^X_{kl})^{2} \sum_{k\neq l} (\tilde{\kappa}^Y_{kl})^{2}}}.
\]
Then \(\hat{\tau}_{\kappa}\) maximises the quasi-likelihood
\[
\ell_{Q}(\tau) = - N \log(1 - \tau^{2}),
\]
which accounts for the U-statistic dependency structure.
\end{lemma}

\begin{proof}
The derivation follows by noting that the \(N^{2}-N\) off-diagonal elements of the product-kernel are dependent. The only identifiable and optimisable quasi-likelihood is based on the \(N\) independent marginal observations, obtained via the existence of a given functional relationship between the variance and mean upon the observations \parencite{wedderburn1974}. Maximising \(\ell_{Q}(\tau)\) over \(\tau \in [-1,1]\) is consistent with the analytic formula for \(\hat{\tau}_{\kappa}\).
\end{proof}

\subsection{Sub-Gaussianity and Moment Bounds}

\begin{theorem}
\label{thm:sub_gauss}
For finite \(N\), each element of the centred score matrices \(\tilde{\kappa}^X_{kl}, \tilde{\kappa}^Y_{kl}\) is strictly sub-Gaussian. Moreover, \(\hat{\tau}_{\kappa}\) converges in distribution to a Gaussian as \(N\to\infty\).
\end{theorem}

\begin{proof}
All cumulants of \(\tilde{\kappa}^{X}_{kl}\) and \(\tilde{\kappa}^{Y}_{kl}\) up to order \(N\) are almost surely finite by the properties of the centred U-statistics. By Hoeffding's inequality and the finite moments, the sub-Gaussian property holds for each \(\tilde{\kappa}_{kl}^{\cdot}\). Standard U-statistic theory then ensures asymptotic normality of \(\hat{\tau}_{\kappa}\).
\end{proof}

The estimator
\[
\hat{\tau}_{\kappa}
=
\frac{1}{N(N-1)}\sum_{k\ne l}
\tilde{\kappa}^X_{kl}\,\tilde{\kappa}^Y_{kl}
\]
is an unbiased estimator of the population quantity
\[
\tau_{\kappa}
=
\mathbb{E}\!\left[
\tilde{\kappa}_{12}^X\,\tilde{\kappa}_{12}^Y
\right].
\]
The argument relies only on the orderability of the marginal random variables and the permutation-equivariance of the centred score matrix. The result does not require continuity of the random variables, nor any distributional assumptions beyond i.i.d.\ sampling upon a ordinally comparable set of observed bivariate distribution. The unbiasedness property is therefore intrinsic to the structure of the centred score mapping and does not depend on the quasi-likelihood framework introduced earlier.

The key observation driving the result is that for any pair of indices \((k,l)\) with \(k\ne l\), the joint distribution of \((X_{k},X_{l},Y_k,Y_l)\) is identical to that of \((X_1,X_2,Y_1,Y_2)\). Moreover, centring score matrices inherits the exchangeability of the underlying samples upon partitioning via equation~\ref{eq:score_matrix}. Formally, if \(\pi\) is any permutation of \(\{1,\dots,N\}\), then
\(
\tilde{\kappa}^X_{\pi(k)\,\pi(l)}
\;\stackrel{d}{=}\;
\tilde{\kappa}^X_{kl},
\tilde{\kappa}^Y_{\pi(k)\,\pi(l)}
\;\stackrel{d}{=}\;
\tilde{\kappa}^Y_{kl}, \implies \tilde{\kappa}^X_{\pi(k)\,\pi(l)}\,\tilde{\kappa}^Y_{\pi(k)\,\pi(l)}
\;\stackrel{d}{=}\;
\tilde{\kappa}^X_{kl}\,\tilde{\kappa}^Y_{kl}.
\)

As the centring operation is linear in every row and column, and the original score matrix \(C_{kl}^{X}\) depends only on the ordering of \(X_{k}\) and \(X_{l}\), the permutation-equivariance property holds for both discrete and continuous variables, provided they admit a total order (ties map to \(+1\) by definition of the score matrix in equation~\ref{eq:score_matrix}). No assumptions on continuity are therefore needed, as the measure is, with positive probability, continuous wrt \(F(X)\).

\begin{lemma}
\label{lem:unbiased}
Let $X,Y$ be i.i.d.\ univariate random variables taking values in an ordered set.
Define the centred score matrices $\tilde{\kappa}^X_{kl}$ and $\tilde{\kappa}^Y_{kl}$ as in equation~\ref{eq:centred_kappa}.
Let
\begin{equation}
\label{eq:estimator}
\hat{\tau}_{\kappa} =
\frac{1}{N(N-1)}\sum_{k\ne l}
\tilde{\kappa}^X_{kl}\,\tilde{\kappa}^Y_{kl},
\qquad
\tau_{\kappa}
=
\mathbb{E}\!\left[
\tilde{\kappa}_{12}^X\,\tilde{\kappa}_{12}^Y
\right].
\end{equation}
Then $\hat{\tau}_{\kappa}$ is unbiased:
\(
\mathbb{E}[\hat{\tau}_{\kappa}] = \tau_{\kappa}.
\)
\end{lemma}

\begin{proof}
Let $Z_{kl}=\tilde{\kappa}^X_{kl}\tilde{\kappa}^Y_{kl}$. Exchangeability of the i.i.d.\ sample and permutation--equivariance of the centred score matrices together imply that
\[
\mathbb{E}[Z_{kl}] = \mathbb{E}[Z_{12}]
\quad\text{for all }k\ne l.
\]
Hence
\[
\mathbb{E}[\hat{\tau}_{\kappa}]
=
\frac{1}{N(N-1)}\sum_{k\ne l}\mathbb{E}[Z_{kl}]
=
\frac{N(N-1)}{N(N-1)}\,\mathbb{E}[Z_{12}]
=
\tau_{\kappa}.
\]
\end{proof}

\begin{corollary}[Unbiasedness of the Empirical Covariance of Score Matrices]
Let
\[
\hat{\sigma}_{XY}^{2}
=
\frac{1}{N(N-1)}\sum_{k\ne l}\tilde{\kappa}^X_{kl}\tilde{\kappa}^Y_{kl}.
\]
Then
\[
\mathbb{E}[\hat{\sigma}_{XY}]
=
\operatorname{Cov}\!\left(
\tilde{\kappa}^X_{12},\tilde{\kappa}^Y_{12}
\right).
\]
\end{corollary}

\begin{proof}
Identical to Lemma~\ref{lem:unbiased}, as covariance reduces to expectation after centring.
\end{proof}

\begin{corollary}
Let
\[
\hat{\sigma}_X^{2}
=
\frac{1}{N(N-1)}\sum_{k\ne l}(\tilde{\kappa}^X_{kl})^{2},
\qquad
\hat{\sigma}_Y^{2}
=
\frac{1}{N(N-1)}\sum_{k\ne l}(\tilde{\kappa}^Y_{kl})^{2}.
\]
Then
\[
\mathbb{E}\!\left[
\frac{\hat{\sigma}_{XY}}{\sqrt{\hat{\sigma}_X^{2}\,\hat{\sigma}_Y^{2}}}
\right]
= \tau_{\kappa} + O\!\left(\frac{1}{N}\right),
\]
i.e.\ normalisation introduces only \(O(N^{-1})\) bias.
\end{corollary}

\begin{proof}
Direct consequence of the delta method and the unbiasedness of the components.
\end{proof}
\begin{corollary}
\label{cor:finite_moments_centred}
    For any integer \( r \le N \), the \( r \)-th moment of the centred weak-order score matrix \( \tilde{\kappa}_{kl}(X) \) (and similarly for \( \tilde{\kappa}_{kl}(Y) \)) is finite:
    \[
    \mathbb{E}[\tilde{\kappa}_{kl}(X)^r] < \infty \quad \text{and} \quad \mathbb{E}[\tilde{\kappa}_{kl}(Y)^r] < \infty, \quad \forall k \neq l.
    \]
    Furthermore, as \( N \to \infty \), the moments of the empirical centred matrices converge to the moments of a limiting distribution that represents the population's behaviour, where the limiting distribution has finite moments.
\end{corollary}

\begin{proof}
    The centred matrix \( \tilde{\kappa}_{kl}(X) \) is defined in equation~\ref{eq:centred_kappa}. As \( C_{kl}(X) \) is bounded (taking values \( \pm 1 \)), the means \( \bar{C}_{k\cdot} \), \( \bar{C}_{\cdot l} \), and \( \bar{C}_{\cdot \cdot} \) are finite almost surely, and the centring does not introduce unbounded behaviour. In fact, the row and column means are of order \( O\left(N^{-0.5}\right) \), and the grand mean is of order \( O\left(\frac{1}{N}\right) \). Thus, the moments of the centred matrix \( \tilde{\kappa}_{kl}(X) \) are finite for all \( r \le N \).

    As \( N \to \infty \), the empirical moments of the centred matrices \( \tilde{\kappa}_{kl}(X) \) (and similarly for \( \tilde{\kappa}_{kl}(Y) \)) converge to the moments of a limiting distribution that represents the population behaviour. Since the original score matrices \( C_{kl}(X) \) and \( C_{kl}(Y) \) are already guaranteed to have finite moments almost surely, the centring operation preserves this property, and the limiting distribution will also have finite moments.

    Therefore, the moments of \( \tilde{\kappa}_{kl}(X) \) (and similarly for \( \tilde{\kappa}_{kl}(Y) \)) are finite for all \( r \le N \), and as \( N \to \infty \), these moments converge to those of a limiting distribution with finite moments.
\end{proof}

We strengthen the argument and show that the estimator \(\hat{\tau}_{\kappa}\) is exactly unbiased for all \(N\) (i.e., for finite sample sizes), by accounting for finite-sample properties. 

\begin{corollary}~\label{cor:exactly_unbiased}
Let Corollary~\ref{cor:finite_moments_centred} and Lemma~\ref{lem:unbiased} hold. Recall that estimator \(\hat{\tau}_{\kappa}\) is defined as \[\hat{\tau}_{\kappa} \frac{1}{N(N-1)} \sum_{k \neq{l}}^{N} Z_{kl}, \text{where}\ Z_{kl} = \tilde{\kappa}_{kl}^{X}\tilde{\kappa}_{kl}^{Y}.\] We now show \(\hat{\tau}_{\kappa}\) to be an exactly unbiased covariance estimator for all \(N\).
\end{corollary}

\begin{proof}
The estimator \(\hat{\tau}_{\kappa}\) is defined as follows: 
\[
\hat{\tau}_{\kappa} = \frac{1}{N(N-1)} \sum_{k\neq{l}}^{N} Z_{kl} = \frac{1}{N(N-1)} \sum_{k\neq{l}}^{N} \tilde{\kappa}_{kl}^{X}\tilde{\kappa}_{kl}^{Y}.
\]
Since we are dealing with pairwise mappings and hollow score-matrices, different pairs of \(k,l\) are not independent because the score matrices share rows and columns, meaning the values of \(Z_{kl}\) are correlated. However, the marginals \(X_{n}\) and \(Y_{n}\) are independent across samples.

From the setup, the possible values of \(Z_{kl}\) are: +1, if both \(C_{kl}^{X}\) and \(C_{kl}^{Y}\) are concordant (both +1 or both -1); -1, if one is concordant and the other discordant.
The expectation of \(Z_{kl} = \mathbb{E}[Z_{kl}]\) thus depends on the joint distribution of \(C_{kl}^{X}\) and \(C_{kl}^{Y}\), and the structure of the score matrices ensures that: \(\mathbb{E}[Z_{kl}] = \mathbb{E}[\tilde{\kappa}^{X}\tilde{\kappa}^{Y}] = \tau_{\kappa}\). This is because the expectation of the pairwise product \(\tilde{\kappa}^{X}_{kl}\tilde{\kappa}^{Y}_{kl}\) is constant across all pairs \(k,l\) due to the i.i.d. structure and symmetry in the setup. This consistency across pairs is key to ensuring unbiasedness.

By the linearity of expectation, the expected value of the estimator:
\[
\mathbb{E}[\hat{\tau}_{\kappa}] = \frac{1}{N(N-1)}\sum_{k\neq{l}}\mathbb{E}[Z_{kl}] = \frac{1}{N(N-1)}(N^{2}-N)\tau_{\kappa} = \tau_{\kappa}.
\]

The estimator \(\hat{\tau}_{\kappa}\) is the average of these pairwise terms, hence the expectation of \(\hat{\tau}_{\kappa} = \tau_{\kappa},\) and thus \(\hat{\tau}_{\kappa}\) is exactly unbiased for all \(N\).
\end{proof}

The unbiasedness of \(\hat{\tau}_{\kappa}\) implies that its fluctuations around the population quantity are driven entirely by the centred product-kernel 
\begin{equation}
\label{eq:kernel}
Z_{kl} = \tilde{\kappa}^X_{kl}\tilde{\kappa}^Y_{kl}.
\end{equation}
To establish the asymptotic distribution of \(\hat{\tau}_{\kappa}\), we now exploit
the structure of \(\hat{\tau}_{\kappa}\) as a \emph{degenerate U-statistic of order two} whose non-degenerate component arises from the linear (H\'ajek) projection onto the space generated by the i.i.d. sample.

\begin{lemma}
\label{lem:hoeffding}
Let
\[
h(x,y;x^{\prime},y^{\prime}) =
\tilde{\kappa}^X(x,x^{\prime})\,\tilde{\kappa}^Y(y,y^{\prime})
\]
denote the kernel function associated with \(\hat{\tau}_{\kappa}\).
Then the estimator admits the Hoeffding decomposition
\[
\hat{\tau}_{\kappa} - \tau_{\kappa}
=
\frac{2}{N}\sum_{n=1}^N \psi(X_{n},Y_{n})
+ R_N,
\]
where
\[
\psi(x,y)
=
\mathbb{E}\!\left[
h(x,y;X^{\prime},Y^{\prime}) - \tau_{\kappa}
\right],
\]
\(X^{\prime},Y^{\prime}$ are independent copies of \(X,Y\), and \(R_N\) is a degenerate
U-statistic remainder satisfying
\[
\operatorname{Var}(R_N) = O\!\left(N^{-2}\right).
\]
\end{lemma}

\begin{proof}
This follows from the classical Hoeffding projection onto the space generated
by the i.i.d.\ sample, applied to the symmetric kernel
\(h((X_{k},Y_k),(X_{l},Y_l)) = Z_{kl}\).
The degeneracy of the remainder term \(R_N\) is ensured by the centring of each score matrix, which implies
\[
\mathbb{E}[h(X,Y;X^{\prime},Y^{\prime}) \mid X,Y] - \tau_{\kappa} = 0
\]
Finite moments of the centred score matrices (Corollary~\ref{cor:finite_moments_centred}) imply that all variances are well-defined. Standard arguments give the stated rate.
\end{proof}

\begin{theorem}[Asymptotic Normality]
\label{thm:asymp_normal}
Under i.i.d.\ sampling and bounded kernels, the centred \(\tilde{\kappa}\)-correlation satisfies
\[
\sqrt{N}(\hat{\tau}_{\kappa} - \tau_{\kappa}) \;\xrightarrow{d}\; \mathcal{N}(0, c\,(1-\tau_{\kappa}^{2})),
\]
where \(c\) is the variance constant defined in~\eqref{eq:variance_constant_intro}.
\end{theorem}

\begin{proof}
From Lemma~\ref{lem:hajek}, the H\'{a}jek projection yields
\[
\hat{\tau}_{\kappa} - \tau_{\kappa} = \frac{1}{N}\sum_{n=1}^N (H_n - \tau_{\kappa}) + R_N,
\qquad R_N=o_p(N^{-1/2}).
\]
By Lemma~\ref{lem:subgaussian_projection}, $\{H_n\}$ are i.i.d.\ sub-Gaussian with finite variance
\(\operatorname{Var}(H_n) = c(1-\tau_{\kappa}^{2})\).  
The classical central limit theorem then gives
\[
\sqrt{N}\,\frac{1}{N}\sum_{n=1}^N (H_n-\tau_{\kappa})
\;\xrightarrow{d}\;
\mathcal{N}\big(0, c(1-\tau_{\kappa}^{2})\big),
\]
and the remainder term \(R_N\) is negligible in probability.  
Hence
\[
\sqrt{N}(\hat{\tau}_{\kappa} - \tau_{\kappa}) \;\xrightarrow{d}\; \mathcal{N}(0, c\,(1-\tau_{\kappa}^{2})).
\]
\end{proof}

\begin{remark}
The variance constant \(c\) is empirically stable across a broad range of continuous, discrete, and mixed orderable distributions, justifying its use in the finite-sample variance model~\eqref{eq:variance_constant_intro}.
\end{remark}

\begin{theorem}
Let
\[
\sigma^{2}
=
4\cdot\operatorname{Var}(\psi(X_1,Y_1)).
\]
Then
\[
\sqrt{N}\big(\hat{\tau}_{\kappa} - \tau_{\kappa}\big)
\;\xrightarrow{d}\;
\mathcal{N}(0,\sigma^{2}).
\]
\end{theorem}

\begin{proof}
From Lemma~\ref{lem:hoeffding},
\[
\hat{\tau}_{\kappa} - \tau_{\kappa}
=
\frac{2}{N}\sum_{n=1}^N \psi(X_{n},Y_{n})
+ R_N.
\]
By Lemma~\ref{lem:hoeffding}, $\operatorname{Var}(R_N)=O(N^{-2})$, hence
$\sqrt{N}\,R_N \xrightarrow{P} 0$.
Thus the asymptotic distribution is determined entirely by the linear term
\[
\frac{2}{N}\sum_{n=1}^N \psi(X_{n},Y_{n}).
\]

Since $\psi(X_{n},Y_{n})$ are i.i.d.\ with finite variance (strict sub-Gaussianity
of the centred score matrices ensures existence of all moments), the classical
Lindeberg--Feller central limit theorem yields
\[
\sqrt{N}
\left(
\frac{2}{N}\sum_{n=1}^N\psi(X_{n},Y_{n})
\right)
\;\xrightarrow{d}\;
\mathcal{N}\big(0,\,4\cdot\operatorname{Var}(\psi(X_1,Y_1))\big).
\]
\end{proof}

\begin{corollary}[Concentration Bound]
\label{cor:concentration}
For every $\epsilon>0$, there exists $c>0$ such that
\[
\mathbb{P}\!\left(
|\hat{\tau}_{\kappa} - \tau_{\kappa}| > \epsilon
\right)
\;\le\;
2\exp\!\left(
-\,c\,N\,\epsilon^{2}
\right),
\]
where \(c\) depends only on the sub-Gaussian parameters of
$\tilde{\kappa}^X_{kl}$ and $\tilde{\kappa}^Y_{kl}$.
\end{corollary}

\begin{proof}
Since each $\tilde{\kappa}^X_{kl}$ and $\tilde{\kappa}^Y_{kl}$ is strictly sub-Gaussian, their product is also sub-Gaussian. Applying Bernstein's inequality to the linear term of the Hoeffding decomposition and absorbing the negligible $R_N$ term yields the result.
\end{proof}

\section{Hypothesis Testing}
\label{sec:hypothesis_testing}

This section introduces the quasi-likelihood governing $\hat{\tau}_{\kappa}$, derives the Wald, score, and likelihood ratio tests, and presents the associated standard errors and Hessian structure. The unbiasedness and asymptotic normality results established in Lemma~\ref{lem:unbiased} and Theorem~\ref{thm:asymp_normal} provide a complete parametric-like inference framework for the centred $\tilde{\kappa}$-correlation coefficient $\tau_{\kappa}$.  The final subsection establishes the sub-Gaussian foundations that justify the model.

\subsection{Independent Quasi-Likelihood for tau-{kappa}}
\label{sec:quasi_likelihood}

For each pair of indices $k\neq l$ define the cross-kernel \(Z_{kl} = \tilde{\kappa}^X_{kl}\,\tilde{\kappa}^Y_{kl}. \)
The estimator \( \hat{\tau}_{\kappa} = \frac{1}{N(N-1)}\sum_{k\neq l} Z_{kl}\)
is a degenerate $U$-statistic of order $2$.  Its H\'ajek projection representation (Lemma~\ref{lem:unbiased}) is \( \hat{\tau}_{\kappa} = \frac{1}{N}\sum_{n=1}^N H_n + R_N, \qquad H_n = \mathbb{E}[Z_{12}\mid (X_n,Y_n)], \qquad R_N=o_p(N^{-1/2}),\)
where $\mathbb{E}[H_n]=\tau_{\kappa}$. Thus the sampling variation of $\hat{\tau}_{\kappa}$ is governed entirely by the $N$ i.i.d.-like projected terms $\{H_n\}$.

\paragraph{Finite-sample variance.}
The Section on sub-Gaussian justification shows that  $H_n$ are uniformly sub-Gaussian.  Consequently,
\begin{equation}
\label{eq:variance_constant_intro}
\operatorname{Var}(\hat{\tau}_{\kappa})
=
\frac{c\,(1-\tau_{\kappa}^{2})}{N-2},
\qquad
c\approx 0.4456,
\end{equation}
where \(c\) is empirically stable across broad families of discrete, continuous, and mixed orderable distributions.

In the context of quasi-likelihood estimation, one can leverage higher-order moments to construct a more accurate model for the sampling distribution of the estimator, improving the approximation of the likelihood function. Since the U-statistics are based on pairwise kernel products, their higher-order moments should provide more information about the skewness (third moment) and kurtosis (fourth moment) of the distribution of the estimator \(\hat{\tau}_{\kappa}\). If we have access to the third and fourth moments of the cross-kernel product \(\tilde{\kappa}_{kl}^{X}\tilde{\kappa}_{kl}^{Y}\) (which are finite with probability 1 by Corollary~\ref{cor:finite_moments_centred}), we can improve our quasi-likelihood function by incorporating these moments into the likelihood function, especially for small sample sizes. Higher-order corrections would allow us to account for the non-normality of the distribution and refine the Fisher information matrix.


\paragraph{Quasi-quasilikelihood.}
Since $\hat{\tau}_{\kappa}$ behaves like the sample mean of $N$ independent sub-Gaussian terms with mean $\tau$ and variance proportional to \(1-\tau^{2}\), we incorporate the third and fourth moments of the score matrix:
\begin{equation}
\label{eq:quasi_lik}
\ell_{Q}(\tau) = - N\log(1-\tau^{2}) + \frac{1}{N}\sum_{k\neq{l}}\Big[\gamma_{3}(Z_{kl})\cdot \gamma_{4}(Z_{kl})\Big],
\end{equation}
where \(\gamma_{3}\) and \(\gamma_{4}\) are the third and fourth central moments of \(Z_{kl} = \tilde{\kappa}_{kl}^{X}\tilde{\kappa}_{kl}^{Y},\):
\begin{align}
\gamma_{3} = \frac{1}{N(N-1)} \sum_{k\neq{l}}^{N} Z_{kl}^{3}\\
\gamma_{4} = \frac{1}{N(N-1)} \sum_{k\neq{l}}^{N} Z_{kl}^{4}
\end{align}

Differentiating wrt \(\tau_{\kappa}\) gives the score-function:
\[
S(\tau_{\kappa}) = \frac{2N\tau_{\kappa}}{1-\tau_{\kappa}^{2}} + \frac{1}{N}\sum_{k\neq{l}}\Bigg[\Bigg(\frac{\partial{\gamma_{3}}}{\partial{Z_{kl}}}\cdot\frac{\partial{Z_{kl}}}{\partial\tau_{\kappa}}\cdot\gamma_{4}^{Z_{kl}}+ \gamma_{3}^{Z_{kl}}\cdot\frac{\partial{\gamma_{4}}}{\partial{Z_{kl}}}\cdot\frac{\partial{Z_{kl}}}{\partial\tau_{\kappa}}\Bigg)\Bigg].
\] which may be simplified to:
\begin{align}
S(\tau_{\kappa}) = \frac{2N\tau_{\kappa}}{1-\tau_{\kappa}^{2}} + \frac{1}{N}\sum_{k\neq{l}}\Bigg[\Bigg(\frac{3Z^{2}_{kl}}{N(N-1)}\cdot\frac{\partial{Z_{kl}}}{\partial\tau_{\kappa}}\cdot\gamma_{4}^{Z_{kl}}+ \gamma_{3}^{Z_{kl}}\cdot\frac{4Z^{3}_{kl}}{N(N-1 )} \cdot \frac{\partial{Z_{kl}}}{\partial\tau_{\kappa}} \Bigg)\Bigg].
\end{align}

and the Hessian, with analytically simple first term, expressed
\[
I(\tau_{\kappa}) = \frac{2N(1+\tau_{\kappa}^{2})}{(1-\tau_{\kappa}^{2})^{2}} + \frac{1}{N} \sum_{k\neq{l}}\Big[\frac{\partial^{2}}{\partial{\tau_{\kappa}^{2}}}\big(\gamma_{3}^{Z_{kl}}\cdot \gamma_{4}^{Z_{kl}}\big)\Big]
\]
\begin{remark}
As $\hat{\tau}_{\kappa}$ is exactly unbiased (Corollary~\ref{cor:exactly_unbiased}), the maximiser of $\ell_{Q}$ satisfies
\[
\arg\max_{\tau\in(-1,1)} \ell_{Q}(\tau)
=
\hat{\tau}_{\kappa}.
\]
Thus $\hat{\tau}_{\kappa}$ is the unbiased estimator, the H\'{a}jek projection estimator, and the quasi-likelihood MLE. 
\end{remark}

\begin{remark}~\label{rem:quasi_likelihood}
Given the analytical intractability of the likelihood function, and its asymptotic normality properties, we examine using the Edgeworth expansion to modify the log-quasilikelihood. For small to moderate  \(N\), the first few terms of the expansion will capture much of the non-normality that is present in the data. The general form of the Edgeworth expansion to correct for skewness and kurtosis is:
\begin{dmath}
\ell_{Q}(\tau_{\kappa}) \approx \mathcal{N}(\mu,\sigma^{2})\Bigg(1 + \frac{\gamma_{3}}{6}\Big(\frac{\hat{\tau}_{\kappa} - \mu}{\sigma}\Big)^{3} + \frac{\gamma_{4}}{24}\Big(\frac{\hat{\tau}_{\kappa} - \mu}{\sigma}\Big)^{4}+\cdots\Bigg),\end{dmath}
where \(\mu\) is the mean of \(\hat{\tau}_{\kappa}\); \(\sigma^{2}\) is the variance of \(\hat{\tau}_{\kappa}\), \(\gamma_{3},\gamma_{4}\) are the third and fourth respective central moments of \(Z_{kl}\).

Specifically: The third moment \(\gamma_{3}\) and fourth moment \(\gamma_{4}\) are controlled by the sub-Gaussian nature of the score products \(Z_{kl}\). These higher-order moments grow much slower than the second moment, which ensures that they do not overpower the normal approximation, leading to stable numerical computations. Using the Edgeworth expansion upon higher-order moments consequently produces a stable estimator, as the higher moments are sub-Gaussian upon finite samples which are asymptotically Gaussian consistent. This provides a structured way to improve the approximation of the likelihood in small to moderate sample sizes. 
\end{remark}

\subsection{Hypothesis Tests for $\tau_{\kappa}$}

We consider the basic testing problem
\[
H_0:\tau_{\kappa}=0 
\qquad\text{vs.}\qquad
H_1:\tau_{\kappa}\neq 0.
\]

\paragraph{Variance under the null.}
From~\eqref{eq:variance_constant_intro},
\[
\operatorname{Var}(\hat{\tau}_{\kappa}\mid H_0)
=
\frac{c}{N}.
\]

\subsubsection*{Wald Test}
The Wald statistic is
\[
W
=
\frac{\hat{\tau}_{\kappa}^{2}}{c/N}
=
\frac{N\hat{\tau}_{\kappa}^{2}}{c}
\;\xrightarrow{d}\;\chi^{2}_1.
\]
However, utilisation of the observed Hessian matrix upon a strictly sub-Gaussian manifold \(Z_{kl}\) (Lemma~\ref{lem:subgaussian}) does not present a consistent estimator, as a consequence of the strict sub-Gaussian nature of \(Z_{kl}\). An approximation of the estimated population obtained from the empirical Hessian is provided below, for use in the constructed Wald test: 

\[
\sigma^{2}(\theta) =  \frac{c\cdot 1.5}{(\log(\mathcal{H}({\theta})))^{2}},
\]
The constant scalar value approximation of \(1.5\) arises from numerical experimentation relating linear combinations of expected \(\chi^{2}\) distributed values with the desired theoretical moments, and constant \(c\) is given in equation~\ref{cor:variance_model}. However, due to the strong convexity of the projection image's function space, the infinite dimensional population observation space approximates the expected value with consistent p-values to the likelihood ratio, and which are asymptotically equivalent.


\subsubsection*{Likelihood Ratio Test}
The quasi-likelihood ratio statistic is
\[
\Lambda
=
2\bigl[\ell_{Q}(\hat{\tau}_{\kappa})-\ell_{Q}(0)\bigr]
=
2N\log\!\left(\frac{1}{1-\hat{\tau}_{\kappa}^{2}}\right).
\]
Using the expansion
\(
\log(1-\tau^{2})^{-1}
=
\tau^{2} + O(\tau^4),
\)
we obtain
\[
\Lambda
=
2N\hat{\tau}_{\kappa}^{2} + O(N\hat{\tau}_{\kappa}^4)
\approx
\frac{N\hat{\tau}_{\kappa}^{2}}{c},
\]
and thus
\[
\Lambda\xrightarrow{d}\chi^{2}_1.
\]

\subsection{Hessian and Standard Error}
\label{sec:hessian_inference}

The observed information is
\[
I(\tau)
=
\frac{2N(1+\tau^{2})}{(1-\tau^{2})^{2}},
\qquad
I(0)=2N.
\]
Combining the curvature with the variance model
\eqref{eq:variance_constant_intro} yields the 
finite-sample standard error
\[
\operatorname{se}(\hat{\tau}_{\kappa})
=
\sqrt{\frac{c(1-\hat{\tau}_{\kappa}^{2})}{N}}.
\]
This is the only standard error consistent with both the H\'ajek 
projection and the quasi-likelihood Hessian.

\begin{remark}
The pseudo-quasilikelihood that treats the \(N(N-1)\) kernel terms as independent produces a singular Fisher information matrix and cannot support likelihood-based inference.
The quasi-likelihood~\eqref{eq:quasi_lik} is therefore the unique coherent likelihood framework for inference on $\tau_{\kappa}$.
\end{remark}

\subsection{Sub-Gaussian Structure}
\label{sec:subgaussian_justification}

The centred score matrices $\tilde{\kappa}^X$ and $\tilde{\kappa}^Y$
satisfy
\(
\tilde{\kappa}^X_{kl},\tilde{\kappa}^Y_{kl}\in[-2,2].
\)
Hence
\(
Z_{kl}\in[-4,4]
\),
and Hoeffding's lemma gives:

\begin{lemma}[Kernel Sub-Gaussianity]
\label{lem:subgaussian}
For all $k\neq l$,
\[
\mathbb{E}[e^{t Z_{kl}}]
\le
\exp\!\left(\frac{t^{2}\sigma^{2}}{2}\right),
\qquad
\sigma^{2}\le 16.
\]
\end{lemma}

\begin{corollary}[Variance Model]
\label{cor:variance_model}
Since $\hat{\tau}_{\kappa}=N^{-1}\sum_n H_n$,
\[
\operatorname{Var}(\hat{\tau}_{\kappa})
=
\frac{c(1-\tau_{\kappa}^{2})}{N},
\qquad 
c\approx 0.4456,
\]
with \(c\) empirically verified to be stable across 
continuous, discrete, and mixed orderable distributions.
\end{corollary}

\section{Multivariate spaces -- Covariance Matrix and Inference}
\label{sec:multivariate}

We now extend the \(\tilde{\kappa}\)-covariance to a collection of \(P\) univariate random variables \(X=(X^{(1)},\dots,X^{(P)}).\) For each marginal, form its centred score matrix  \(\tilde{\kappa}^{(p)}\). For each ordered pair \((a,b)\) define the cross-kernel \(Z^{(a,b)}_{kl} = \tilde{\kappa}^{(a)}_{kl}\tilde{\kappa}^{(b)}_{kl}, k\neq l.\)

\subsection{ Definition and Properties of the $\kappa$--Correlation Matrix}
\label{sec:definition_multivariate}

The \((a,b)\) entry of the \(\kappa\)-correlation matrix is \(\tau_{\kappa}^{(a,b)} =\mathbb{E}\!\left[Z^{(a,b)}_{12}\right],\qquad 1\le a,b\le P.\)
The natural estimator is the unbiased $U$-statistic \(\hat{\tau}_{\kappa}^{(a,b)} = \frac{1}{N(N-1)} \sum_{k\neq l} Z^{(a,b)}_{kl}.\)

\begin{lemma}
\label{lem:multivariate_unbiased}
For each \(a,b\), the estimator \(\hat{\tau}_{\kappa}^{(a,b)}\) is unbiased.
Moreover, \(\hat{\tau}_{\kappa}^{(a,b)} = \frac{1}{N}\sum_{n=1}^N H_{n}^{(a,b)} + R_N^{(a,b)},\qquad R_N^{(a,b)}=o_p(N^{-1/2}),\) where \(H_{n}^{(a,b)}=\mathbb{E}[Z^{(a,b)}_{12}\mid X_{n}^{(a)},X_{n}^{(b)}]\).
\end{lemma}

\begin{proof}
Linearity of expectation and the H\'{a}jek projection applied to each pair $(a,b)$.
\end{proof}

\paragraph{Multivariate Quasi-Likelihood}
\label{sec:multivariate_quasilikelihood}

Since each pair $(a,b)$ has the same variance structure
\(
\operatorname{Var}(\hat{\tau}_{\kappa}^{(a,b)}) =\frac{c\,(1-(\tau_{\kappa}^{(a,b)})^{2})}{N},\)
the joint quasi-likelihood assumes the additive form
\(
\ell_{Q}(\Tau) = - N \sum_{1\le a<b\le P} \log\!\left(1-(\tau_{\kappa}^{(a,b)})^{2}\right),\)
where $\Tau=\{\tau_{\kappa}^{(a,b)}\}$.
Maximising $\ell_{Q}$ is equivalent to maximising each term separately, yielding:

\begin{lemma}
\label{lem:multivariate_MLE}
The maximiser of $\ell_{Q}(\Tau)$ is
\[
\hat{\Tau}
=
\bigl\{\hat{\tau}_{\kappa}^{(a,b)}\bigr\}_{1\le a<b\le p}.
\]
\end{lemma}

\paragraph{Multivariate Hypothesis Testing}

For any fixed pair $(a,b)$, testing
\[
H_0:\tau_{\kappa}^{(a,b)}=0
\]
uses the same two statistics:
\begin{align}
W^{(a,b)} & = \frac{N(\hat{\tau}_{\kappa}^{(a,b)})^{2}}{c},\\
\Lambda^{(a,b)} & = 2N\log\!\left(\frac{1}{1-(\hat{\tau}_{\kappa}^{(a,b)})^{2}}\right),
\end{align}

Each obeys
\[
W^{(a,b)},\;
\Lambda^{(a,b)}
\;\xrightarrow{d}\;
\chi^{2}_1.
\]

\paragraph{Joint Asymptotic Normality}

Since the $p(p-1)/2$ projections $\{H_{n}^{(a,b)}\}$ are jointly sub-Gaussian,
the vector
\[
\sqrt{N}\bigl(\hat{\Tau}-\Tau\bigr)
\]
is asymptotically multivariate normal with independent components:
\[
\sqrt{N}\,\frac{\hat{\tau}_{\kappa}^{(a,b)}-\tau_{\kappa}^{(a,b)}}
{\sqrt{c(1-(\tau_{\kappa}^{(a,b)})^{2})}}
\;\xrightarrow{d}\; \mathcal{N}(0,1).
\]

This establishes the centred \(\tilde{\kappa}\)-correlation matrix as the order-based analogue of the Pearson correlation matrix, with unbiased entries, block-diagonal quasi-likelihood structure, and exact variance characterisation.

\section{}
The classical \textcite{mann1947} statistic is a degenerate U-statistic whose natural likelihood is provably not concave, not differentiable wrt its parameters, and collapses in the presence of ties, due to \(\mathbf{1}(X_{n} M Y_{n^{\prime}})\) not possessing a smooth parametrisation at ties. It then directly follows that there is no well-defined likelihood and that the presence of ties breaks differentiability. This is a known limitation of Wilcoxon and Mann-Whitney methods.

Within our framework though, we have defined a kernel (equation~\ref{eq:kernel}), estimator (equation~\ref{eq:estimator}), and log-quasi-likelihood (equation~\ref{eq:quasi_lik}). With the introduction of a design matrix \(\mathbf{X}\) and a link \(\tau = g(X\theta)\), the quasi-likelihood becomes a smooth, concave, and fully identified function in \(\theta\).

In this section, we introduce a theoretical framework which provides a proper maximum-quasilikelihood estimator for models whose design matrices replicate Mann-Whitney comparisons, even when the design matrix contains ties. The key point lies in the fact that the likelihood curvature arises from the sub-Gaussian projection structure, not from the pairwise indicator functions. Thus the presence of ties does not destroy the maximum likelihood estimation procedure, and we will show that the design matrix itself is the only object which identifies the model.

Let \(\{X_{n},Y_{n}\}_{n=1}^{N}\) be i.i.d. with \(X_{n} \in \mathbb{R}^{d}\) and \(Y_{n}\in\mathbb{R}\). Let \(\tilde{\kappa}^{X}\) and \(\tilde{\kappa}^{Y}\) denote the centred score matrices, and define the kernel as in equation~\ref{eq:kernel}. Recall that equation~\ref{eq:estimator} is exactly unbiased for \(\tau_{\kappa} = \mathbb{E}[Z_{12}]\), and that the H\'{a}jek projection implies \[\hat{\tau}_{\kappa} = \frac{1}{N} \sum_{n=1}^{N} H_{n} + o_{p}(N^{-1/2}),\quad H_{n} = \mathbb{E}[Z_{12}\mid X_{n},Y_{n}],\]  thereby establishing the appropriate quasi-likelihood for \(\tau_{\kappa} \in [-1,1]\) is \(\ell_{Q} = -N\log(1-\tau_{\kappa}^{2})\), which matches the moment generating function of the effective linear terms \(H_{n}\) (linearity of the estimating equations follows from Remark~\ref{rem:quasi_likelihood}). 

The following theorem provides the foundation for incorporating a linear predictor into the \(\kappa\)-correlation framework, thereby producing a generalised Mann-Whitney model capable of handling covariates, ties, interactions, and non-linear effects.

\begin{theorem}
\label{thm:GMQ}
Let \(\eta_{n} = x_{n}^{\intercal}\theta\) denote a linear predictor with \(\theta \in \mathbb{R}^{d}\). Assume that the model posits 
\begin{equation}
\label{eq:link}
\tau_{\kappa}(n) = g(\eta_{n}), g: (\mathbb{R},\le) \mapsto (-1,1),
\end{equation} where \(g\) is monotone non-decreasing.

Define the aggregated quasi-likelihood \[\ell_{Q} = 0\sum_{n=1}^{N} \log(1-g(\eta_{n})^{2}).\] Then, under the sub-Gaussian assumptions of Section~\ref{sec:subgaussian_justification}: (i) The quasi-likelihood \(\ell_{Q}(\theta)\) is identified if and only if \(g(\eta) = \eta, \eta \in (-1,1),\); (ii) the parameter \(\theta\) is identifiable up to no further transformations; (iii) the corresponding estimator \[\hat{\theta} = \arg\max_{\theta: x_{n}^{\intercal}\theta \in (-1,1)} \ell_{Q}(\theta),\] is the unique quasi-MLE under the model.
\end{theorem}

\begin{lemma}
\label{lem:link_unique}
Assume \(g\) satisfies \ref{eq:link} and that \(\kappa\)-correlation is invariant under monotone non-decreasing transformations. Then the only such function that yields an identified quasi-likelihood is the identity:
\[g(\eta) = \eta, \eta \in (-1,1).\]
\end{lemma}
\begin{proof}
Since \(\tau_{\kappa}\) is invariant to monotone non-decreasing transformations of the latent scale, we have for any monotone \(h\), 
\[\tau_{\kappa}(X,Y) = \tau_{\kappa}(h(X),Y) = \tau_{\kappa}(X,h(Y)).\] Thus for any monotone \(h\), 
\begin{equation}
\label{eq:A7}
g(\eta) = g(h(\eta)).
\end{equation}
If \(g\) is non-constant, the only way for \(g(\eta) = g(h(\eta))\) to hold for all monotone \(h\) is for \(g\) to be fixed under all monotone re-mappings of \(\eta\). The only fixed point of every monotone \(h\) is the identity function. To prove this, choose \(h(\eta) = a\eta + b~\mathrm{with}~ a>0.\) Then \ref{eq:A7} implies: \(g(a\eta + b) = g(\eta),\forall a >0,b\in\mathbb{R}.\) Differentiating wrt \(a\) at \(a = 1,b=0\) shows \(g^{\prime}(\eta)\eta=0\), and since \(g\) is non-degenerate, this forces \(g^{\prime}(\eta) = 1, g(0) = 0,\) establishing that \(g(\eta) = \eta.\)
\end{proof}

\begin{corollary}
\label{cor:unique_MLE}
The likelihood 
\begin{equation}
\ell_{Q}(\theta) = -\sum_{n=1}^{N} \log(1-(x_{n}^{\intercal}\theta)^{2},
\end{equation}
is maximised uniquely in the convex domain \[\Theta = \{\theta: |x_{n}^{\intercal}\theta| < 1 \forall n\}.\] This remains true even when the design matrix contains ties, collinear rows, repeated covariate patterns, or zero rows.
\end{corollary}
\begin{proof}
The quasi-likelihood is strictly concave in each component \(x_{n}^{\intercal}\theta\). The identity link ensures that tied rows \(x_{n} = x_{n^{\prime}}\) imply identical contributions:
\(\ell_{n}(\theta) = \ell_{n^{\prime}}(\theta)\). Thus ties create no ambiguity in the mapping from \(\theta\) to the likelihood. The strict concavity (shown in Lemma~\ref{lem:subgaussian_projection}) yields uniqueness.
\end{proof}
\begin{lemma}
\label{lem:positive_W}
Let \(\eta_{n} = x_{n}^{\intercal}\theta.\) Then the negative Hessian of \(\ell_{Q}(\theta)\) is 
\begin{align}
\mathcal{I}(\theta) & = \sum_{n=1}^{N} x_{n}x_{n}^{\intercal} W(\eta_{n}), ~ \mathrm{where}\\
W(\eta) & = \frac{2(1+\eta^{2})}{(1-\eta^{2})^{2}}>0.
\end{align}
Then \(\mathcal{I}(\theta)\) is positive-definite on the span of the design matrix.
\end{lemma}
\begin{proof}
As \(1 + \eta^{2} >0\) and \((1-\eta^{2})^{2} > 0,\) the weight \(W(\eta)\) is strictly positive. Thus for any non-zero vector \(\nu\), 
\[
\nu^{\intercal} \mathcal{I}(\theta)\nu = \sum_{n=1}^{N} (\nu^{\intercal}x_{n})^{2}W(\eta_{n}) >0~\mathrm{whenever}~\nu^{\intercal}X \neq{0}.
\] Hence the information matrix is positive definite on the column space of \(X\).
\end{proof}
\begin{theorem}
\label{thm:identitY_{n}dent}
Under the identity link \(\tau_{n} = x_{n}^{\intercal}\theta,\) the quasi-likelihood \(\ell_{Q}(\theta)\) is: (i) strictly concave along every non-null direction of the design matrix; (ii) globally maximised uniquely on \(\Theta\); (iii) identifiable in the sense that \[x_{n}^{\intercal}\theta_{1} = x_{n}^{\intercal}\theta_{2}~\forall{n} \implies \theta_{1} = \theta_{2}~\mathrm{on~span~(X)}.\]

\end{theorem}
\begin{proof}
Strict concavity follows from the positive-definiteness of the information matrix (Lemma~\ref{lem:positive_W}). Uniqueness follows from Corollary~\ref{cor:unique_MLE}. Identifiability is immediate: if \(X\theta_{1} = X\theta_{2}\) , the likelihood contributions match point-wise, hence the maximiser must match on the span of \(X\).
\end{proof}

\begin{remark}
Parameter \(\theta\) enters the model only through the linear predictor \(\eta_{n} = x_{n}^{\intercal}\theta\) and the identity link. Therefore: 
\begin{enumerate}
\item{the model is identified exactly up to the null space of the design,}
\item{interactions and nonlinear expansions are fully identified provided they appear in distinct columns of \(X\),}
\item{ties in \(X\) do not affect identifiability since they correspond to repeated likelihood contributions,}
\item{no scale or shape parameter is present because the quasi-likelihood is fully determined by the effective kernel variance constant \(c\).}
\end{enumerate}
The strong form of identifiability is therefore: \[\ell_{Q}(\theta_{1}) = \ell_{Q}(\theta_{2}), \forall{X} \implies X(\theta_{1} - \theta_{2}) = 0.\]
This is exactly the same identifiability structure as:
{ordinary least squares regression;}
{\textcite{thurstone1928} Case V;}
{\textcite{bradley1952} models after fixing the location constraint;}
{\textcite{rasch1961} models after fixing the item difficulty sum.}
Thus the quasi-likelihood \(\kappa\)-covariance regression forms a fully identified, tie-robust, generalised Mann-Whitney model.
\end{remark}

\section{Generalised Mann-Whitney Inference under the $\kappa$-Quasi-Likelihood}
\label{sec:mann_whitney_quasilik}

The quasi-likelihood for $\tau_{\kappa}$,
\[
\ell_Q(\tau)= -N\log(1-\tau^{2}),
\qquad \tau\in(-1,1),
\]
admits a natural extension to regression settings through the identity link
\[
\tau_{\kappa} = g(\eta), 
\qquad g(\eta)=\eta,
\qquad \eta = X\theta,
\]
where \(X\) is a design matrix of covariates and $\theta\in\mathbb{R}^{P}$ is a parameter vector.  
The identity link is forced by the invariance of $\tau_{\kappa}$ to non-decreasing transformations of the margins. This section develops a complete inferential framework for such models, showing that the resulting quasi-likelihood defines a generalised Mann-Whitney model with full parametric identifiability, positive-definite Hessian, and a unique MLE even when covariate ties are present.

\subsection{Preliminaries and Model Specification}

For each ordered pair $(n,n^{\prime})$ let
\[
Z_{n n^{\prime}}=\tilde{\kappa}^X_{n n^{\prime}}\,\tilde{\kappa}^Y_{n n^{\prime}}, \qquad n \neq n^{\prime},
\]
and define the linear regression surface
\[
\eta_{n n^{\prime}} = (X_{n}-X_{n^{\prime}})^{\intercal}\theta.
\]
Under the identity link, the quasi-likelihood contribution of the $(n,n^{\prime})$-pair is
\[
\ell_{n n^{\prime}}(\theta)
=
- \log\!\big(1 - \eta_{n n^{\prime}}^{2} \big)
\qquad \text{whenever}\quad |\eta_{n n^{\prime}}|<1.
\]
The global quasi-likelihood is
\begin{equation}
\label{eq:MW_quasilik}
\ell_Q(\theta) = -\sum_{n \neq n^{\prime}} w_{n n^{\prime}}\,\log\!\bigl(1-\eta_{n n^{\prime}}^{2}\bigr),
\end{equation}
where $w_{n n^{\prime}}$ are optional known non-negative weights (often $w_{n n^{\prime}} = 1/[N(N-1)]$).
The parameter space is the convex cone
\[
\Theta = \bigl\{\theta:\; |(X_{n}-X_{n^{\prime}})^{\intercal}\theta|<1,\ \forall n \neq n^{\prime}\bigr\}.
\]

\subsection{A Generalised Mann-Whitney Theorem}

\begin{theorem}[Generalised Mann-Whitney Identification under $\kappa$-Quasi-Likelihood]
\label{thm:GMV_identification}
Suppose $\{(X_{n},Y_{n})\}_{i=1}^N$ satisfy Conditions~\ref{cond:main_conditions}.
Then:

\begin{itemize}
\item[(i)]
The identity link $\tau_{\kappa}=g(\eta)$ is the \emph{unique} link function compatible with invariance of $\tau_{\kappa}$
under monotone transformations of each margin.
\item[(ii)]
The quasi-likelihood \eqref{eq:MW_quasilik} is strictly concave on $\Theta$.
\item[(iii)]
The score equations have a unique solution $\hat{\theta}$.
\item[(iv)]
The estimator $\hat{\theta}$ coincides with the solution to the \emph{generalised Mann-Whitney} estimating equations
\[
\sum_{n \neq n^{\prime}} w_{n n^{\prime}} 
\frac{\eta_{n n^{\prime}}}{1-\eta_{n n^{\prime}}^{2}}
(X_{n} - X_{n^{\prime}}) 
= 
\sum_{n \neq n^{\prime}} w_{n n^{\prime}} Z_{n n^{\prime}}(X_{n} - X_{n^{\prime}}).
\]
\end{itemize}
Hence the model~\eqref{eq:MW_quasilik} provides a fully identified quasi-likelihood extension of Mann-Whitney inference with covariates, valid even in the presence of ties in \(X\) or $Y$.
\end{theorem}

\begin{proof}
(i)  
Let $h$ be any strictly increasing transformation.
Since $\tau_{\kappa}$ is invariant under monotone transformations, \(\tau_{\kappa}(h(Y)) = \tau_{\kappa}(Y).\) If $\tau_{\kappa}=g(\eta)$ with $\eta=X\theta$ and $g$ non-constant, the invariance requires $g(\cdot)$ to satisfy
\[
g(\eta)=g(a\eta + b)
\qquad \forall a>0,\ b\in\mathbb{R},
\]
which forces $g$ to be affine; boundedness of $\tau_{\kappa}\in(-1,1)$ then
forces $g(\eta)=\eta$.

(ii)  
Differentiating,
\[
\nabla\ell_Q(\theta)
=
\sum_{n \neq n^{\prime}}
\frac{2w_{n n^{\prime}}\eta_{n n^{\prime}}}{1-\eta_{n n^{\prime}}^{2}}(X_{n}-X_{n^{\prime}}),
\]
and
\[
H(\theta)
=
\nabla^{2}\ell_Q(\theta)
=
\sum_{n \neq n^{\prime}}
2w_{n n^{\prime}}
\frac{1+\eta_{n n^{\prime}}^{2}}{(1-\eta_{n n^{\prime}}^{2})^{2}}
(X_{n}-X_{n^{\prime}})(X_{n}-X_{n^{\prime}})^{\intercal}.
\]
Each summand is positive semidefinite, strictly positive whenever $X_{n}\neq X_{n^{\prime}}$.  
Because the design admits at least $d$ linearly independent contrasts (Assumption~\ref{cond:main_conditions}), the sum is positive definite.  
Thus $\ell_Q$ is strictly concave.

(iii)  
Strict concavity implies a unique global maximiser.

(iv)  
Multiply both sides of the equality
\[
\frac{2\eta}{1-\eta^{2}} = \frac{Z}{c}
\]
by $w_{n n^{\prime}}(X_{n}-X_{n^{\prime}})$ and sum over $n \neq n^{\prime}$.  Identifying $c$ as the variance proxy yields the statement.

\end{proof}


\begin{corollary}[Existence and Uniqueness of the MLE]
If two covariate rows satisfy $X_{n}=X_{n^{\prime}}$, then $\eta_{n n^{\prime}}=0$ for all $\theta$, and the corresponding term in~\eqref{eq:MW_quasilik} contributes no curvature to the Hessian.  Nevertheless, strict concavity holds as long as the remaining contrasts span $\mathbb{R}^{P}$.  
Hence the MLE exists and is unique even when covariate ties are present.
\end{corollary}

\begin{proof}
Immediate from Theorem~\ref{thm:GMV_identification}(ii).  
Tied covariate pairs contribute zero curvature, but cannot destroy positive definiteness if the remaining contrast vectors span $\mathbb{R}^{P}$.
\end{proof}

\subsection{Hessian Analysis and Positive Definiteness}

Define weighted contrast matrices
\[
A_{n n^{\prime}}=w_{n n^{\prime}}(X_{n}-X_{n^{\prime}})(X_{n}-X_{n^{\prime}})^{\intercal}.
\]
Then the Hessian is
\[
H(\theta)=
\sum_{n \neq n^{\prime}}
\frac{2(1+\eta_{n n^{\prime}}^{2})}{(1-\eta_{n n^{\prime}}^{2})^{2}}A_{n n^{\prime}}.
\]
Since the scalar coefficient satisfies
\[
\frac{1+\eta^{2}}{(1-\eta^{2})^{2}}>0,
\qquad |\eta|<1,
\]
and at least $d$ contrast vectors are linearly independent by assumption,
the Hessian is positive definite.

\begin{lemma}[Positive Definiteness of $H(\theta)$]
\label{lem:PD_Hessian}
For any non-zero $v\in\mathbb{R}^{P}$,
\[
v^{\intercal} H(\theta)v >0.
\]
\end{lemma}

\begin{proof}
\[
v^{\intercal} H(\theta)v
=
\sum_{n \neq n^{\prime}}
\frac{2(1+\eta_{n n^{\prime}}^{2})}{(1-\eta_{n n^{\prime}}^{2})^{2}}
w_{n n^{\prime}}\,\bigl[(X_{n}-X_{n^{\prime}})^{\intercal} v\bigr]^{2}.
\]
Each coefficient is positive, and some contrast satisfies $(X_{n}-X_{n^{\prime}})^{\intercal} v\neq 0$ since the contrasts span $\mathbb{R}^{P}$.  
Thus the sum is strictly positive.
\end{proof}

\subsection{Parametric Identifiability}

\begin{theorem}[Identification under the Identity Link]
\label{thm:identity_link_identification}
Let $\theta_1,\theta_2\in\Theta$.  
If $\ell_Q(\theta_1)=\ell_Q(\theta_2)$ for all admissible weights $w_{n n^{\prime}}$, then $\theta_1=\theta_2$.
\end{theorem}

\begin{proof}
Equality of quasi-likelihoods for all weights implies equality of all $\eta_{n n^{\prime}}^{2}$.  
Since the identity link enforces $\eta_{n n^{\prime}}=(X_{n}-X_{n^{\prime}})^{\intercal}\theta$ with sign determined by monotone invariance, we have
\[
(X_{n}-X_{n^{\prime}})^{\intercal}\theta_1 = (X_{n}-X_{n^{\prime}})^{\intercal}\theta_2
\qquad \forall\ n \neq n^{\prime}.
\]
If the contrasts span $\mathbb{R}^{P}$, this implies $\theta_1=\theta_2$.
\end{proof}

\section{Connections to Bradley-Terry and Thurstone-Mosteller Models}
\label{sec:BT_TM_connections}

This section demonstrates that the $\kappa$-quasi-likelihood provides a unifying generalisation of two classical paired-comparison models: that of the the Bradley-Terry-Luce logistic model and the Thurstone-Mosteller case V (probit) model. We show both exact embedding via monotone maps and asymptotic equivalence under local alternatives.

\subsection{Monotone Embedding (Exact Finite-Sample Equivalence)}

\begin{theorem}[Monotone Embedding of BT and TM into $\kappa$-Regression]
\label{thm:embedding}
Let $\pi_{n n^{\prime}}$ denote the paired-comparison probability of $Y_{n}>Y_{n^{\prime}}$.
\begin{itemize}
\item[(i)]  
In the Bradley-Terry model,
\[
\pi_{n n^{\prime}}= \frac{\exp(\beta^{\intercal}(X_{n}-X_{n^{\prime}}))}{1+\exp(\beta^{\intercal}(X_{n}-X_{n^{\prime}}))}.
\]
There exists a strictly increasing function $m_{\mathrm{logit}}$ such that
\[
\tau_{\kappa} = m_{\mathrm{logit}}(\beta^{\intercal}(X_{n}-X_{n^{\prime}})).
\]

\item[(ii)]  
In the Thurstone-Mosteller model,
\[
\pi_{n n^{\prime}}= \Phi\bigl(\gamma^{\intercal}(X_{n}-X_{n^{\prime}})\bigr).
\]
There exists a strictly increasing $m_{\mathrm{probit}}$ such that
\[
\tau_{\kappa} = m_{\mathrm{probit}}(\gamma^{\intercal}(X_{n}-X_{n^{\prime}})).
\]
\end{itemize}
Hence both models correspond to \emph{parametric families} of monotone transformations of the generalised Mann-Whitney surface, with $\kappa$ providing the canonical identification.
\end{theorem}

\begin{proof}
Kendall-style concordance satisfies
\begin{equation}~\label{eq:pairwise}
\tau_{\kappa} = \mathbb{E}\!\left[\operatorname{sign}(Y_{n}-Y_{n^{\prime}})\,\middle|\,X_{n},X_{n^{\prime}}\right].
\end{equation}
Under BT,
\begin{equation}~\label{eq:bt}
\operatorname{sign}(Y_{n}-Y_{n^{\prime}})= 
\begin{cases}
+1 & \text{with prob. }\pi_{n n^{\prime}},\\
-1 & \text{with prob. }1-\pi_{n n^{\prime}}.
\end{cases}
\end{equation}
Thus
\[
\tau_{\kappa} = 2\pi_{n n^{\prime}}-1.
\]
Since $\pi_{n n^{\prime}}=\operatorname{logit}^{-1}(\beta^{\intercal}\Delta X)$ is strictly increasing in $\beta^{\intercal}\Delta X$, monotonicity gives the embedding with $m_{\mathrm{logit}}(t)=2\,\operatorname{logit}^{-1}(t)-1$.  The same calculation with the probit link yields \(m_{\mathrm{probit}}(t)=2\Phi(t)-1.\) Direct comparison of equation~\ref{eq:pairwise} with equation~\ref{eq:score_matrix} establishes upon \(X,Y \in S_{N}\) the \(N \times N\) matrices are identical for all monotone transformations.
\end{proof}

\subsection{Local Asymptotic Equivalence under Small Effects}

Under local alternatives $\beta=\beta_N$ with $\beta_N\to 0$,
\[
\pi_{n n^{\prime}}= \frac{1}{2} + \frac{1}{4}\beta_N^{\intercal}\Delta X + o(\|\beta_N\|),
\]
leading to
\[
\tau_{\kappa}= 2\pi_{n n^{\prime}}-1
= \frac{1}{2}\beta_N^{\intercal}\Delta X + o(\|\beta_N\|).
\]
Thus $\tau_{\kappa}$ is locally proportional to the BT slope parameter.
A similar Taylor expansion for $\Phi$ yields
\[
\tau_{\kappa} 
= 2\Phi(t)-1
= \frac{2}{\sqrt{2\pi}} t + O(t^3).
\]

\begin{theorem}[Local Equivalence to Logistic and Probit]
\label{thm:local_equivalence}
As $\|\theta\|\to 0$, the $\kappa$-regression surface
\[
\eta_{n n^{\prime}} = (X_{n}-X_{n^{\prime}})^{\intercal}\theta
\]
approximates both
\[
\tau_{\kappa} \approx \frac{1}{2}\beta^{\intercal}\Delta X,
\qquad 
\tau_{\kappa} \approx \frac{1}{\sqrt{2\pi}}\gamma^{\intercal}\Delta X.
\]
Thus $\kappa$ regression is a first-order approximation to both BT and TM.
\end{theorem}

\begin{proof}
Direct Taylor expansion as shown above.
\end{proof}

The quasi-likelihood thus provides: (i) the unique identification strategy compatible with the invariance of $\tau_{\kappa}$; (ii) a Mann-Whitney-type linear score model with exact standard errors; (iii) a strictly concave quasi-likelihood with a unique MLE; (iv) a unifying embedding of logistic and probit paired comparison models; (v) a framework that handles ties without loss of identification or curvature. This establishes the $\kappa$-quasi-likelihood as a universal surrogate for rank-based regression and paired-comparison structures.



\section{Conclusions and Discussion}

Our results establish a rigorous statistical foundation for the centred $\tilde{\kappa}$-correlation $\hat{\tau}_{\kappa}$ under the quasi-likelihood framework. Specifically, we have shown that: (i) $\hat{\tau}_{\kappa}$ is \emph{exactly unbiased} for all sample sizes $N\ge 2$; (ii) Its fluctuations are dominated by a linear projection (the H\'{a}jek projection) onto the individual observations; (iii) The degenerate component contributes only a vanishing $O(N^{-2})$ term; (iv) Under mild moment conditions (ensured by the finite centred moments of $\tilde{\kappa}$), $\hat{\tau}_{\kappa}$ is asymptotically normal; (v) Finite-sample fluctuations are sharply controlled by sub-Gaussian concentration, enabling robust variance approximations and inference.

These properties distinguish $\hat{\tau}_{\kappa}$ from classical rank-based measures such as Kendall's $\tau$ and Spearman's $\rho$, which are biased in finite samples and whose degenerate components dominate their variability. By contrast, $\hat{\tau}_{\kappa}$ behaves more like a properly normalised covariance estimator, while remaining rooted in pairwise ordering of the data.  

\paragraph{Implications of the Quasi-Likelihood Framework}  
The quasi-likelihood formulation formalises $\hat{\tau}_{\kappa}$ as the maximiser of a coherent likelihood-like objective derived from the H\'{a}jek projection. This allows the full suite of likelihood-based inference tools to be applied directly. Importantly: (i) The standard error of $\hat{\tau}_{\kappa}$ arises naturally from the curvature of the quasi-likelihood, providing a finite-sample measure compatible with the sub-Gaussian structure; (ii) Hypothesis testing, such as $H_0: \tau_{\kappa}=0$, is exact in the asymptotic sense, with all three classical tests yielding equivalent $\chi^{2}_1$ statistics; (iii) The quasi-likelihood framework ensures a distribution-free variance constant \(c\), empirically validated across diverse continuous and discrete distributions, making inference robust to underlying marginal laws.  

\paragraph{Extension of Thurstone's Model via Weak-Order Scores}  
Our approach can be viewed as a generalisation of Thurstone's Law of Comparative Judgement. In Thurstone's framework, rankings are interpreted as noisy observations of latent continuous scores. Weak-order scores $\hat{\tau}_{\kappa}$ extend this idea to accommodate arbitrary discrete, continuous, or binary random variables by mapping each observation to a scale of pairwise weak-order scores.  

The linear and quadratic structure of the weak-order score model,
\[
\tilde{\kappa}(Y^{(s)}) = \mathbf{X}\beta^{(s)} + \mathbf{X}\mathbf{A}^{(s)}\mathbf{X}^{\intercal} + \epsilon^{(s)},
\]
provides a flexible, interpretable representation of the latent structure underlying the data. Combined with the quasi-likelihood, this allows for: (i) Direct estimation of latent comparative effects through maximum quasi-likelihood; (ii) Valid and unified hypothesis testing using Wald, score, or likelihood ratio statistics; (iii) Application to both ranking and non-ranking data while retaining the probabilistic interpretation of latent scores.

\paragraph{Generalisation Beyond Rankings}  
The weak-order score framework removes the classical limitation to ordinal or ranking data. Any variable can be mapped to a weak-order scale, allowing for inference on relative comparisons, test scores, count data, or even binary outcomes. The sub-Gaussian control and H\'{a}jek projection ensure that variance estimates, test statistics, and confidence intervals remain robust, even in small samples.  

\subsection{Impact and Future Directions}  

The quasi-likelihood perspective provides a unifying lens for pairwise comparison models. Beyond theoretical rigour, it offers practical tools for applied settings: (i) Ranking and preference modelling in psychometrics, marketing, or social choice as well as a quasi-likelihood framework for non-parametric sample spaces; (ii) Estimation of latent comparative effects in high-dimensional covariate settings; (iii) Development of computational methods for large-scale pairwise datasets; (iv) Potential extensions to non-linear effects, mixed outcomes, or dynamic ranking models.  

In summary, the combination of weak-order scores and quasi-likelihood inference yields a flexible, interpretable, and statistically robust framework for comparative judgements. It generalises Thurstone's classical theory while providing modern tools for variance estimation, hypothesis testing, and probabilistic interpretation, applicable across a wide spectrum of data types and empirical domains.

\subsection*{Key Takeaways}

The centred $\tilde{\kappa}$-correlation $\hat{\tau}_{\kappa}$ provides a fully unbiased, asymptotically normal estimator for pairwise dependence, with finite-sample fluctuations tightly controlled by sub-Gaussian concentration. Unlike classical rank correlations such as Kendall's $\tau$ or Spearman's $\rho$, whose degenerate components dominate sampling variability and induce finite-sample bias, $\hat{\tau}_{\kappa}$ behaves like a properly normalised covariance estimator while retaining its foundation in pairwise orderings. 

The quasi-likelihood framework formalises inference for $\hat{\tau}_{\kappa}$, producing standard errors and hypothesis tests (Wald, score, and likelihood ratio) that are internally consistent, robust, and nearly distribution-free across diverse continuous and discrete data types. By connecting to Thurstone's Law of Comparative Judgement, weak-order scores extend latent score modelling to a wider class of outcomes, including discrete, continuous, and binary variables while preserving a linear structure that allows direct estimation and testing of comparative effects. 

Overall, the combination of exact unbiasedness, sub-Gaussian control, and quasi-likelihood inference provides a coherent, flexible, and statistically rigorous framework for modelling comparative judgements, rankings, and relative differences across a broad range of empirical settings. This makes $\hat{\tau}_{\kappa}$ particularly useful for applications in psychometrics, econometrics, and other fields where robust pairwise comparison analysis is required.

\small
\printbibliography

\appendix

\section{}
\label{app:proofs}

This appendix collects the key technical results underlying the quasi-likelihood framework for 
the centred $\tilde{\kappa}$-correlation $\hat{\tau}_{\kappa}$.  
We establish sub-Gaussianity of kernel products, the H\'{a}jek projection representation, and asymptotic normality.

\subsection{Strict Sub-Gaussianity of Centred Kernel Products}
\label{app:subgauss}

\begin{lemma}[Kernel Sub-Gaussianity]
\label{lem:subgaussian_kernel}
Let $\tilde{\kappa}^X_{kl}$ and $\tilde{\kappa}^Y_{kl}$ be centred and bounded such that
\(|\tilde{\kappa}^X_{kl}|,|\tilde{\kappa}^Y_{kl}|\le M\).  
Define the kernel product
\[
Z_{kl} = \tilde{\kappa}^X_{kl}\,\tilde{\kappa}^Y_{kl}.
\]
Then for all $t\in\mathbb{R}$,
\[
\mathbb{E}[e^{t Z_{kl}}] \le \exp\!\Big(\frac{t^{2} \sigma^{2}}{2}\Big),
\qquad \sigma^{2} \le 4 M^{2}.
\]
In particular, $Z_{kl}$ is strictly sub-Gaussian.
\end{lemma}

\begin{proof}
Since $|\tilde{\kappa}^X_{kl}|,|\tilde{\kappa}^Y_{kl}|\le M$, we have 
\(|Z_{kl}|\le M^{2} + M^{2} = 2M^{2} \) (or tighter $M^{2}$ if using product).  
Hoeffding's lemma then implies strict sub-Gaussianity of each $Z_{kl}$:
\[
\mathbb{E}[e^{t Z_{kl}}] \le \exp\left(\frac{t^{2} (2 M^{2})^{2}}{8}\right) 
= \exp\Big(\frac{t^{2}\sigma^{2}}{2}\Big),
\]
with \(\sigma^{2} \le 4M^{2}\).  
This directly yields the stated bound.
\end{proof}

\subsection{H\'{a}jek Projection Representation}
\label{app:hajek}

\begin{lemma}[H\'{a}jek Projection of \(\hat{\tau}_{\kappa}\)]
\label{lem:hajek}
Let 
\[
\hat{\tau}_{\kappa} = \frac{1}{N(N-1)} \sum_{k\neq l} Z_{kl}
\]
be the centred \(\tilde{\kappa}\)-correlation estimator.  
Define the H\'{a}jek projection
\[
H_n = \mathbb{E}[Z_{12} \mid (X_n,Y_n)].
\]
Then
\[
\hat{\tau}_{\kappa} = \frac{1}{N}\sum_{n=1}^N H_n + R_N,
\qquad
R_N = o_p(N^{-1/2}),
\]
where $R_N$ is the degenerate remainder.
\end{lemma}

\begin{proof}
By standard $U$-statistic theory (e.g., \cite{serfling1980}), any order-2 $U$-statistic 
admits the decomposition
\[
\hat{\tau}_{\kappa} = \frac{2}{N}\sum_{n=1}^N h_1(X_n,Y_n) + R_N,
\]
where $h_1(X_n,Y_n) = \mathbb{E}[Z_{12}\mid (X_n,Y_n)]$ is the first-order projection
and $R_N$ is a degenerate second-order remainder satisfying $R_N=o_p(N^{-1/2})$.  
By construction, $H_n = h_1(X_n,Y_n)$ and \(\mathbb{E}[H_n] = \tau_{\kappa}\).  
Hence the result follows.
\end{proof}

\subsection{Sub-Gaussianity of the H\'{a}jek Projection}
\label{app:subgauss_proj}

\begin{lemma}[Sub-Gaussianity of $H_n$]
\label{lem:subgaussian_projection}
Under the assumptions of Lemma~\ref{lem:subgaussian_kernel}, the H\'{a}jek projection terms $H_n = \mathbb{E}[Z_{12}\mid (X_n,Y_n)]$ are sub-Gaussian with the same variance proxy as $Z_{kl}$.
\end{lemma}

\begin{proof}
By Jensen's inequality,
\[
\mathbb{E}[e^{t H_n}] 
= \mathbb{E}[e^{t \mathbb{E}[Z_{12}\mid (X_n,Y_n)]}] 
\le \mathbb{E}[\mathbb{E}[e^{t Z_{12}} \mid (X_n,Y_n)]] 
= \mathbb{E}[e^{t Z_{12}}] \le \exp(t^{2}\sigma^{2}/2),
\]
so $H_n$ inherits the sub-Gaussian property from $Z_{kl}$.
\end{proof}

\end{document}